\documentclass[submission,copyright,creativecommons]{eptcs}
\providecommand{\event}{AiML 2026}

\usepackage{aiml26}
\usepackage[noadjust]{cite}

\usepackage{iftex}
\ifpdf
  \usepackage{underscore}         
  \usepackage[T1]{fontenc}        
\else
  \usepackage{breakurl}           
\fi
\usepackage{bm}

\usepackage{graphicx} 
\let\oldalpha\alpha
\renewcommand{\alpha}{\scalebox{0.85}{$\oldalpha$}}

\usepackage{amsmath, amssymb}
\usepackage{centernot}
\usepackage{esvect}
\usepackage{semantic, syntax}
\usepackage[inline]{enumitem}
\usepackage{hyperref}
\usepackage{caption}
\usepackage{subcaption}
\usepackage{pgfplots}
\usepgfplotslibrary{fillbetween}
\pgfplotsset{compat=1.18}
\usepackage{multicol}
\usepackage{import}

\usepackage{tikz}
\usetikzlibrary{3d,
    patterns,
    patterns.meta,
    arrows,
    arrows.meta,
    positioning,
    shapes,
    calc,
    intersections,
    hobby,
    decorations.markings
}
\usepackage{bussproofs}

\usepackage[bibliography=common,appendix=append]{apxproof}

\theoremstyle{plain}
\newtheoremrep{theorem}{Theorem}[section]
\newtheoremrep{proposition}[theorem]{Proposition}
\newtheoremrep{corollary}[theorem]{Corollary}
\newtheoremrep{lemma}[theorem]{Lemma}

\usepackage{totcount}
\newtotcounter{todonum}

\newcommand{\shelf}[1]{}
\newcommand{\change}[1]{\textcolor{purple}{#1}}
\renewcommand{\change}[1]{#1}

\title{Some Results on Causal Modalities in General Spacetimes}
\author{Marco Lewis \qquad\qquad Nesta van der Schaaf
\institute{Université Paris-Saclay, CNRS, CentraleSupélec, ENS Paris-Saclay, Inria,\\ Laboratoire Méthodes Formelles, 91190, Gif-sur-Yvette, France}}

\newcommand{\titlerunning}{Causal Modalities in General Spacetimes}
\newcommand{\authorrunning}{M. Lewis, N. van der Schaaf}
\hypersetup{
  bookmarksnumbered,
  pdftitle    = {\titlerunning},
  pdfauthor   = {\authorrunning},
  pdfsubject  = {EPTCS},               
  pdfkeywords = {}
}

\def\etc{\emph{etc}.}
\def\ie{\emph{i.e.},}
\def\eg{\emph{e.g.},}

\newcommand{\norm}[1]{\left\lVert#1\right\rVert}
\DeclareMathOperator{\im}{im}

\newcommand{\genericworld}{M}

\newcommand{\genericrelation}{\mathrel{\triangleleft}}

\newcommand{\genericeval}{V}
\newcommand{\atomicprops}{AP}
\newcommand{\mframe}[2]{\langle #1, #2 \rangle}
\newcommand{\genericframe}{\mframe{\genericworld}{\genericrelation}}
\newcommand{\mmodel}[3]{\langle #1, #2, #3 \rangle}
\newcommand{\genericmodel}{\mmodel{\genericworld}{\genericrelation}{\genericeval}}

\newcommand{\entail}[3]{#1\if\relax\detokenize{#2}\relax\else,\fi #2 \models #3}
\newcommand{\nentail}[3]{#1\if\relax\detokenize{#2}\relax\else,\fi #2 \not\models #3}

\newcommand{\mentail}[3]{\entail{#1}{#2}{#3}}
\newcommand{\nmentail}[3]{\nentail{#1}{#2}{#3}}
\newcommand{\mentailgeneric}[1]{\mentail{\genericmodel}{x}{#1}}

\newcommand{\chron}{\ll}
\newcommand{\chroneq}{\mathrel{\underline{\chron}}}
\newcommand{\horismos}{\rightarrow}
\newcommand{\caus}{\preccurlyeq}
\newcommand{\after}{\mathrel{\alpha}}

\newcommand{\necc}{\square}
\newcommand{\poss}{\lozenge}

\newcommand{\inev}{\raise.4ex\hbox{$\bigtriangledown$}} 

\newcommand{\modal}[1]{\mathcal{L}(#1)}

\newcommand{\afterformula}{a \alpha f}
\newcommand{\aftertwoformula}{a \alpha_2 f}
\newcommand{\mK}{\mathbf{K}}
\newcommand{\afterlogic}{\mathbf{D4da}}
\newcommand{\aftertwologic}{\mathbf{D4da_2}}
\newcommand{\gabb}{\text{(Gabb)}}

\tikzset{empty/.style={draw=none},
    refl/.style={fill=black},
    every label/.style={draw=none,inner sep=0},
}
\newcommand{\tikzminkowski}{\def\xmax{2}
          \draw[->,thick] (0,-\xmax) -- (0,\xmax+0.2) node[left=-.5] {$t$};
          \draw[->,thick] (-\xmax,0) -- (\xmax+0.2,0) node[below=0] {$x$};
}
 
\begin{document}

\maketitle

\begin{abstract}
    Causality is one of the fundamental structures of spacetimes, determining the possible behaviour and propagation of physical information. 
    Causal structure can be analysed through the various modal logics it induces.
    The modal logics for the chronological and causal relations of the archetypal Minkowski spacetime have been classified.
    However, only partial results have been achieved for the strict variant of the causal relation, known as the \emph{after relation}.
    Towards classification, it was shown by Shapirovsky and Shehtman that the after modality in Minkowski space satisfies a formula we call the `after formula'.

    The present work continues this analysis towards arbitrary spacetimes.
    In particular, we prove that the after modality in any smooth spacetime satisfies the after formula.
    We introduce a related modal formula that demonstrates that the logic of two-dimensional spacetimes are more expressive than higher-dimensional ones.
    Lastly, we study the interrelation between the logical properties and physical properties along the causal ladder.
\end{abstract}

\section{Introduction}
There has been some research into the modal logic of \emph{Minkowski space}, the spacetime of special relativity~(Example~\ref{example:minkowski space}).
One of the early works is by Goldblatt~\cite{Goldblatt1980}, who showed that the modal logic of any~$1+n$-dimensional Minkowski spacetime with the causal relation~$\caus$, \change{moving up to and possibly at the speed of light}, is~$\mathbf{S4.2}$.
Shehtman~\cite{Shehtman1983} also observed similar results around that time for various~$\mathbf{S4}$ logics.
Shapirovsky and Shehtman~\cite{Shapirovsky2002} showed that the chronological relation $\chron$, \change{moving strictly slower than the speed of light}, again for any $1+n$-dimensional Minkowski spacetime, is given by the modal logic $\mathbf{OI.2}$, which is a sublogic of $\mathbf{S4.2}$. \change{The general definitions of~$\caus,\chron$ are given in~Section~\ref{sec:spacetime}.}

Recent works, such as \cite{Hirsch2018, Hirsch2022}, have been investigating the temporal modal logic, using both past and future modalities, of the strict variant of the causal relation (what is referred to as the \emph{after} relation; see Section~\ref{sec:spacetime:causalstruct}).
However, it has only been shown in two-dimensional Minkowski spacetime that the temporal logic of the after relation is decidable \cite{Hirsch2018}.
Whilst the frames of Minkowski spacetime with the chronological and causal relation are complete with respect to their logics, it is unknown what the modal logic of Minkowski spacetime with the after relation is, and whether it is complete, although previous works have noted it has distinct properties~\cite{Goldblatt1980, Shapirovsky2005}. Mainly, it was proved that the modal logic of the after relation satisfies a formula we call the \emph{after formula}~(Section~\ref{sec:afterformula}).

To the best of our knowledge, the modal logics of general spacetimes has not been investigated.\footnote{Though there are related approaches such as via domain theory~\cite{martin2006DomainSpacetimeIntervals} or orthomodular lattices~\cite{casini2002LogicCausallyClosed}.}
What is the modal logic of a general spacetime?
Some properties can be derived from the general differential geometric setup, such as that $\chron$ is transitive and 2-dense ($\mathbf{OI}$) and $\caus$ is a preorder ($\mathbf{S4}$).
One of the main difficulties in providing a more refined characterisation lies in the fact that, even though spacetimes locally behave like Minkowski space (locally inheriting the logic), general spacetimes can vary greatly in global topology and causal structure. Our main result~(Theorem~\ref{thm:spacetimeafter}) proves that the after relation of \emph{any} spacetime satisfies the after formula. Moreover, we study how the modal logics of the spacetime vary along where it sits on the so-called causal ladder~(Sections~\ref{sec:causalladder} and~\ref{sec:modalcausalladder}).


\section{Spacetimes and Causal Structure}
\label{sec:spacetime}\label{sec:spacetime:causalstruct}
We recall the basic definitions of \emph{causal structure} in smooth spacetimes~\cite{penrose1972TechniquesDifferentialTopology, landsman2021FoundationsGeneralRelativity}.
Formally, a \emph{spacetime} is a (real second countable Hausdorff) smooth manifold~$M$ equipped with a \emph{Lorentzian metric}~$g$ and a \emph{time orientation}~$T$. The metric defines at each point~$x\in M$ a nondegenerate symmetric bilinear form~$g_x\colon T_xM\times T_xM\to \mathbb{R}$ on the tangent space, to be thought of as a generalised inner product. The difference is that~$g_x$ need not be positive definite, meaning non-zero vectors may have zero or even negative length. The Lorentzian signature is interpreted as one \emph{time dimension} and $n$ \emph{spatial dimensions}. The time orientation $T$ is a global smooth vector field that determines a consistent notion of future.

The fundamental \emph{causal relations} are defined in terms of certain classes of curves induced by the metric (see Figure~\ref{fig:mink2}). \change{For our purposes a \emph{curve} is a continuous (piecewise) smooth function~$\gamma\colon I\to M$ defined on an interval~$I\subseteq \mathbb{R}$, and its \emph{image} is the set~${\im(\gamma)= \{\gamma(t):t\in I\}}$.} For $x,y\in M$ we define the \emph{causal(ity)}, \emph{chronology} and \emph{horismos relation} as follows, respectively:
\begin{itemize}
    \item $x\caus y$ iff there exists a causal curve (\emph{moving up to lightspeed}) from $x$ to $y$, or $x=y$;
    \item $x\chron y$ iff there exists a timelike curve (\emph{moving slower than lightspeed}) from $x$ to $y$;
    \item $x\horismos y$ iff there exists a lightlike curve (\emph{moving at lightspeed}) from $x$ to $y$, or $x=y$ (\ie{}~$x\caus y$ and not $x \chron y$).
\end{itemize}
\change{For more detailed definitions see~\cite[\S 1.11]{minguzzi2019LorentzianCausalityTheory}.}
Additionally, we define the \emph{after} (or \emph{strictly causal}) and \emph{reflexive chronological relation} as follows, respectively:
\begin{itemize}
    \item $x \after y$ iff there exists a causal curve from $x$ to $y$ (read as ``$y$ happens after $x$'');
    \item $x \chroneq y$ iff there exists a timelike curve from $x$ to $y$, or $x=y$.
\end{itemize}

The after relation, introduced in \cite{robb1914TheoryTimeSpace}, is the `strict' version of $\caus$, since causal curves are by definition non-degenerate, only allowing a reflexive pair $x\after x$ when there exists a causal loop from $x$ to $x$.
We say~$x$ and~$y$ are \emph{spacelike separated} if not $x\caus y$ and not $y\caus x$.

\begin{example}[Minkowski spacetime]\label{example:minkowski space}
An important spacetime is $1+n$-dimensional \emph{Minkowski space}, based on the manifold $\mathbb{R}^{1+n}$. We denote a point $x\in \mathbb{R}^{1+n}$ in components as ${(x^0,\vec{x})=(x^0,x^1,\ldots,x^n)}$. The metric and induced causal relations are then as follows (see~Figure~\ref{fig:mink2} for intuition):
\begin{equation}
    \eta(x,y) = -x^0y^0 + \sum_{i=1}^n x^iy^i,
    \quad\text{and}\quad
    \begin{array}{cc}
        x \chron y 
    & \quad \text{iff}\quad
    y^0 > x^0 + \norm{\vec{y}-\vec{x}},
    \\
    x \caus y
    &\quad\text{iff}\quad
    y^0 \geq x^0 + \norm{\vec{y}-\vec{x}},
    \\
    x \horismos y
    &\quad\text{iff}\quad
    y^0 = x^0 + \norm{\vec{y}-\vec{x}}.
    \end{array}
    \label{eq:minkowski causal relations}
\end{equation}
\end{example}

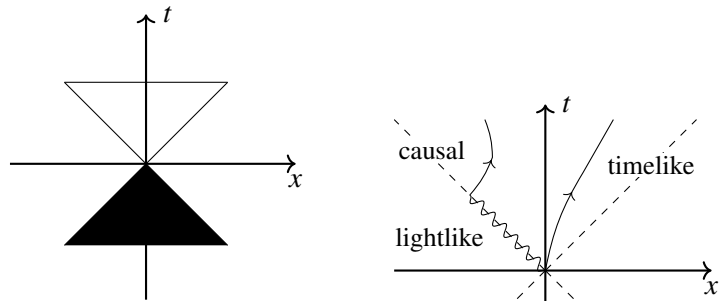
\begin{figure}[bh]
    \centering
    \begin{subfigure}[t]{.3\textwidth}
        \centering
        \begin{tikzpicture}[scale=.9]
        \def\x{0}
        \def\s{1.2}
        \fill[white] (\x,0) -- (\x+\s,\s) -- (\x-\s,\s);
        \fill[black] (\x,0) -- (\x+\s,-\s) -- (\x-\s,-\s);
        \draw (\x,0) -- (\x+\s,\s);
        \draw (\x,0) -- (\x-\s,\s);
        \draw (\x-\s,\s) -- (\x+\s,\s);
        \draw (\x,0) -- (\x+\s,-\s);
        \draw (\x,0) -- (\x-\s,-\s);
        \tikzminkowski;
        \end{tikzpicture}
    \end{subfigure}
    \hspace{5mm}
    \begin{subfigure}[t]{.3\textwidth}
        \centering
        \begin{tikzpicture}
        \def\x{0}
        \def\s{2}
        \def\diag{dashed}
        \clip (-2,-.4) rectangle + (5,3);
        \fill[white] (\x,0) -- (\x+\s,\s) -- (\x-\s,\s);
        \fill[white] (\x,0) -- (\x+\s,-\s) -- (\x-\s,-\s);
        \draw[\diag] (\x,0) -- (\x+\s,\s);
        \draw[\diag] (\x,0) -- (\x-\s,\s);
        \draw[\diag] (\x,0) -- (\x+\s,-\s);
        \draw[\diag] (\x,0) -- (\x-\s,-\s);
        \node [] (y) at (0.5, 1.3) {};
        \draw[postaction={decorate},decoration={markings,mark=at position 0.5 with {\arrow{>}}}] (\x,0) to [bend left=10] (y.center) to (.9,2);
        \node [] (z) at (-0.7, 1.5) {};
        \node [] (l) at (-1, 1) {};
        \draw[decorate, decoration={snake, segment length=6pt, amplitude=2pt}] (\x,0) to (l.center);
        \draw[postaction={decorate},decoration={markings,mark=at position 0.5 with {\arrow{>}}}] (l.center) to [bend right = 10] (z.center) to [bend right=10] (-.8,2);
        \tikzminkowski;
        \node at (-1.4,.4) {\small lightlike};
        \node [fill=white,inner sep=1pt] at (1.4,1.4) {\small timelike};
        \node [fill=white,inner sep=1pt] at (-1.5,1.6) {\small causal};
        \end{tikzpicture}
    \end{subfigure}
    \caption{Two-dimensional Minkowski spacetime and its causal structure.
    \change{Timelike curves move strictly inside of the cone, with tangent strictly greater than~$45^\circ$ from the horizontal axis at every point. Lightlike curves move precisely at~$45^\circ$, while causal curves are a mixture of the two.}}
    \label{fig:mink2}
\end{figure}

\subsection{Properties of the causal relations}
A central theme of this work is characterising the (modal) properties of the causal relations of a spacetime. We first collect some well-established results from the literature~\cite{penrose1972TechniquesDifferentialTopology, kronheimer1967StructureCausalSpaces}.

\begin{definition}
    A relation ${\genericrelation} \subseteq S \times S$ is called \emph{semi-full} if:
    \begin{itemize}
        \item (seriality) $\forall x \in S$, there exists $y \in S$ such that $x \genericrelation y$;
        \item (2-dense) if $x \genericrelation y_1, y_2$, there exists $z \in S$ such that $x \genericrelation z \genericrelation y_1, y_2$.
    \end{itemize}
\end{definition}

\begin{lemma}\label{lemma:relations between causal orders}\label{lem:relcommonproperties}
    We have the following relations and properties of the causal orders:
    \begin{multicols}{2}
    \begin{itemize}
        \item $x\chron y \implies x\after y \implies x\caus y$;
        \item if $x\neq y$ then $x\after y$ iff $x\chron y$ or $x\horismos y$;
        \item[]
    \end{itemize}
\columnbreak
    \begin{itemize}
        \item $\chron$ is semi-full, $\after$ is serial;
        \item $\chron$, $\chroneq$, $\after$, $\caus$ are transitive;
        \item $\caus$, $\chroneq,\horismos$ are reflexive. 
    \end{itemize}
\end{multicols}
\end{lemma}



\begin{proposition}[Push-up Rule]
We have that for $x, y, z \in \genericworld$:
\[
x\chron y \caus z \implies x\chron z
\qquad\text{and}\qquad
x\caus y \chron z \implies x\chron z.
\]
\end{proposition}

\subsection{The Causal Ladder}
\label{sec:causalladder}
\begin{figure}[t]
    \centering
    \begin{tikzpicture}[
        every path/.style={implies-, double equal sign distance}, 
        every node/.style={}, 
        node distance = 3mm and 5mm,
    ]
        \node[] (ntv) {Non-totally Vicious};
        \node[below = of ntv, minimum width = 3cm] (bntv) {};
        \node[right = of bntv, minimum width = 2.5cm] (chron) {Chronological};
        \node[left = of bntv, minimum width = 2.5cm] (cntv) {\textbf{Causally Non-totally Vicious (cNTV)}};
        \node[below = of bntv] (cntvchron) {\textbf{cNTV and Chronological}};
        \node[below = of cntvchron] (caus) {Causal};
        \node[below = of caus] (dist) {(Past-/Future-) Distinguishing};
        
        \draw (ntv.south east) -- (chron.north west);
        \draw (ntv.south west) -- (cntv.north east);
        \draw (chron.south west) -- ([shift={(1.6,0)}]cntvchron.north);
        \draw (cntv.south east) -- ([shift={(-1.6,0)}]cntvchron.north);
        \draw (cntvchron) -- (caus);
                
        \node[below = of dist] (strongly) {Strongly Causal};
        \node[below = of strongly] (stably) {Stably Causal};
        \node[below = of stably] (gh) {Globally Hyperbolic};
        \draw (caus) -- (dist);
        \draw (dist) -- (strongly);
        \draw (strongly) -- (stably);
        \draw (stably) -- (gh);
    \end{tikzpicture}
    \caption{Modified Causal Ladder (our additions in bold).}
    \label{fig:causalladder}
\end{figure}
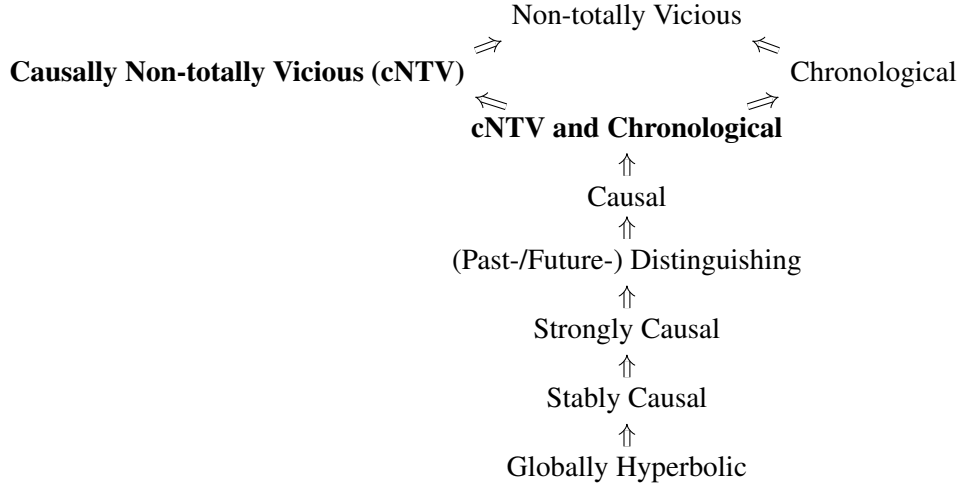

The causal ladder~\cite{minguzzi2019LorentzianCausalityTheory} provides a classification of spacetimes based on the properties of their causal relations. A basic version of it, sufficient for the present work, is depicted in Figure~\ref{fig:causalladder} (with the weakest property at the top).
Below we define the properties needed in this work (see Section~\ref{sec:modalcausalladder}); for more disambiguated versions see \eg{}~\cite[Figure 8]{minguzzi2019LorentzianCausalityTheory}. Let $J^\pm$ and $I^\pm$ denote the up/down closures of $\caus$ and~$\chron$, respectively. So $J^{+}(x):=\{y\in M:x\caus y\}$ and $J^{-}(x):= \{ z\in M: z\caus x\}$, and similarly for $I^\pm$.

\begin{definition}
    A spacetime is called:
    \begin{itemize}
        \item \emph{totally vicious} if $x\chron x$ for all $x\in M$ (\ie{}~$\chron$ is \emph{reflexive}).
        A spacetime is \emph{non-totally vicious} (NTV) if there is some $x \in M$ such that not $x \chron x$ (\ie{}~$\chron$ is not reflexive);
        
        \item \emph{chronological} if not $x \chron x$ for all $x \in M$ (\ie{}~$\chron$ is \emph{irreflexive});

        \item \emph{causal} if $x \caus y$ and $y \caus x$ imply $x = y$ (\ie{}~$\caus$ is \emph{anti-symmetric});

        \item \emph{past-distinguishing} (resp.~\emph{future-distinguishing}) if $I^{-}(x) = I^{-}(y)$ (resp.~$I^{+}(x) = I^{+}(y)$) implies $x = y$. If a spacetime is both past- and future-distinguishing it is called \emph{distinguishing};

        \item \emph{globally hyperbolic} if the \emph{causal diamonds} $J^{+}(x)\cap J^{-}(y)$ are compact in the manifold topology for all $x,y\in \genericworld$.
    \end{itemize}
\end{definition}

For proofs of implication between the different properties, see \cite{Minguzzi08,minguzzi2019LorentzianCausalityTheory}.
We note that the implications are strict, \eg{} there are chronological spacetimes that are not causal.
The study of the causal ladder and the underlying dependencies is one of the cornerstones of modern mathematical relativity theory.

We provide two examples of spacetime that reside on different levels of the causal ladder, but we will see other examples throughout the paper.

\begin{example}
    Minkowski space is globally hyperbolic.
\end{example}

\begin{example}\label{ex:Misner spacetime}
The spacetime depicted in Figure~\ref{fig:cntvchron:cntv} is a non-totally vicious spacetime.
The manifold is $\mathbb{R} \times S^1$, the surface of an infinitely long unit cylinder, and the lightcones are tilted from $90^\circ$ to being upright at $0^\circ$. This is known as the \emph{Misner spacetime}, see~\cite[Example 4.28]{minguzzi2019LorentzianCausalityTheory} for details. It is non-totally vicious because points on the dotted line, $(0, \theta)$, cannot chronologically reach themselves (even though causally they can).
\end{example}

\subsection{Causally Non-Totally Vicious: A New Property for the Causal Ladder}
\label{sec:spacetime:cntv}
Causal ladder properties can have multiple and equivalent definitions.
In particular, the causal step has an interesting equivalent condition.

\begin{lemma}
    Let $\genericworld$ be a spacetime.
    Then $\genericworld$ is causal, \ie{}~$\caus$ is anti-symmetric, iff $\after$ is irreflexive.
\end{lemma}

Both the chronological and causal properties on the causal ladder are affecting the spacetime in a similar way, but affect different relations.
With the chronological relation, $\chron$, there is a difference between losing reflexivity (being non-totally vicious) and gaining irreflexivity (being chronological).
However, the causal property in the ladder combines the loss of reflexivity and gaining of irreflexivity on $\alpha$ into one ladder property.
By changing the perspective of the causal property affecting $\after$, we can introduce a class of spacetimes that separates the notion of losing reflexivity and gaining irreflexivity.
Thus, we get a causal analogue of non-total vicious.

\begin{definition}
A spacetime is \emph{causally non-totally vicious (cNTV)} iff $\exists x \in \genericworld$ such that not $x \after x$.
\end{definition}

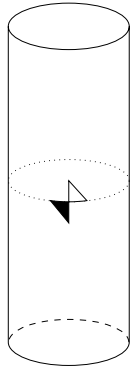
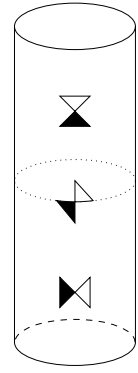
\begin{figure}[t]
    \centering
    \begin{subfigure}[t]{.45\textwidth}
        \centering
        \scalebox{.8}{
\begin{tikzpicture}
    \node[cylinder, draw, shape border rotate = 90, aspect=3, minimum height = 6cm, minimum width=2cm] (A) {};
    \draw[dashed]
    let \p1 = ($ (A.after bottom) - (A.before bottom) $),
        \n1 = {0.5*veclen(\x1,\y1)-\pgflinewidth},
        \p2 = ($ (A.bottom) - (A.after bottom)!.5!(A.before bottom) $),
        \n2 = {veclen(\x2,\y2)-\pgflinewidth}
  in
    ([xshift=-\pgflinewidth] A.before bottom) arc [start angle=0, end angle=180,
    x radius=\n1, y radius=\n2];

    \node[ellipse,draw,dotted,minimum width=2cm, minimum height=.7cm] at (0,.35) {};
    \newcommand{\lightconedistance}{.35}
    \coordinate (p) {};
    \draw[fill=white] (p) arc (-90:-60:.6 and .15) -- (0,\lightconedistance) -- (p);
    \draw[fill] (p) arc (-90:-120.5:.6 and .15) -- (0,-\lightconedistance) -- (p);
    
\end{tikzpicture}
}
        \caption{A chronological (non-causal) spacetime that is not cNTV.
        Every point acts lightlike on a slice of the cylinder allowing every point to (causally) loop around to itself, but a point cannot reach itself in a timelike way.}
        \label{fig:cntvchron:chron}
    \end{subfigure}
    \hfill
    \begin{subfigure}[t]{.45\textwidth}
        \centering
        \scalebox{.8}{
\begin{tikzpicture}









\node[cylinder, draw, shape border rotate = 90, aspect=3, minimum height = 6cm, minimum width=2cm] (A) {};
\draw[dashed]
let \p1 = ($ (A.after bottom) - (A.before bottom) $),
    \n1 = {0.5*veclen(\x1,\y1)-\pgflinewidth},
    \p2 = ($ (A.bottom) - (A.after bottom)!.5!(A.before bottom) $),
    \n2 = {veclen(\x2,\y2)-\pgflinewidth}
in
([xshift=-\pgflinewidth] A.before bottom) arc [start angle=0, end angle=180,
x radius=\n1, y radius=\n2];

\node[ellipse,draw,dotted,minimum width=2cm, minimum height=.7cm] at (0,.35) {};

\newcommand{\lightconedistance}{.35}
\coordinate (p) {};
\draw[fill=white] (p) arc (-90:-60:.6 and .15) -- (0,\lightconedistance) -- (p);
\draw[fill] (p) arc (-90:-120.5:.6 and .15) -- (0,-\lightconedistance) -- (p);

\coordinate (fp) at (0,1.5);
\coordinate (pp) at (0,-1.5);

\draw[fill=white] (fp) -- ([shift={(.25, .25)}]fp) -- ([shift={(-.25, .25)}]fp) -- (fp);
\draw[fill] (fp) -- ([shift={(.25,-.25)}]fp) -- ([shift={(-.25,-.25)}]fp) -- (fp);

\draw[fill=white] (pp) -- ([shift={(.25, .25)}]pp) -- ([shift={(.25,-.25)}]pp) -- (pp);
\draw[fill] (pp) -- ([shift={(-.25, .25)}]pp) -- ([shift={(-.25,-.25)}]pp) -- (pp);
\end{tikzpicture}
}
        \caption{A cNTV (non-causal) spacetime that is not chronological~\cite[Fig. 10]{minguzzi2019LorentzianCausalityTheory}.
        The points above the boundary are irreflexive in $\alpha$ but the points below the boundary are reflexive in $\after$ and $\chron$.}
        \label{fig:cntvchron:cntv}
    \end{subfigure}
        \caption{Spacetimes for Lemma~\ref{lem:cntvnotchron}.}
    \label{fig:cntvchron}
\end{figure}

The cNTV property makes $\alpha$ lose reflexivity and now the causal property only gives $\alpha$ the irreflexive property, separating the relational properties into individual steps on the causal ladder.
There is an obvious question to ask: where does cNTV sit on the causal ladder?
\begin{proposition}
    A causal spacetime is cNTV, 
    and a cNTV spacetime is NTV. 
    \label{prop:causcntvntv}
\end{proposition}
The proof of these implications is simply based on the definitions of the properties and is obvious.

Then, what is the relation between a spacetime being cNTV and it being chronological?
\begin{theorem}
    We have that
    \begin{enumerate}[label=(\roman*)]
        \item there are chronological spacetimes that are not cNTV (Figure~\ref{fig:cntvchron:chron});
        \item there are cNTV spacetimes that are not chronological (Figure~\ref{fig:cntvchron:cntv}).
    \end{enumerate}
    \label{lem:cntvnotchron}
\end{theorem}
\begin{proof}
The proof of both can be observed in the spacetimes depicted in Figure~\ref{fig:cntvchron}, where each has one property but not the other. Both spacetimes use the same manifold, $\mathbb{R} \times S^{1}$. The metric of Figure~\ref{fig:cntvchron:chron} is just quotiented Minkowski space, and Figure~\ref{fig:cntvchron:cntv} is Misner spacetime from Example~\ref{ex:Misner spacetime}.
\end{proof}


        

\begin{remark}
\label{rmk:cntvchron}
It is clear that a causal spacetime is cNTV and chronological, but we further have that a spacetime that is cNTV and chronological is not necessarily causal.
A counterexample can be observed by taking the spacetime in Figure~\ref{fig:cntvchron:chron} and removing a point from it.
\end{remark}

The cNTV property changes one end of the causal ladder to look different to the original ladder.
Our additions to the causal ladder in Figure~\ref{fig:causalladder} are in bold.
In Section~\ref{sec:modalcausalladder}, we will see how introducing the cNTV property into the causal ladder will make a clearer distinction of the modal logic of a spacetime to its position on the causal ladder.

\section{Modal Logic}
\label{sec:modal}
Now we introduce modal logic, which includes the standard logical connectives (conjuction, disjunction, \etc), atomic propositions, and modal operators ($\necc, \poss$).
Throughout, we assume a finite set of atomic propositions, $\atomicprops = \{p_1, p_2, \dots\}$.
We refer to \cite{Blackburn2001,BoxesAndDiamonds} for common definitions and theorems.
Throughout, we make use of Kripke frames, $\genericframe$, using $\genericworld$ to denote the set of points/worlds and $\genericrelation$ to denote the relation.
For models we use $\genericeval : \atomicprops \to 2^{\genericworld}$ to denote the evaluation function.
For a relation, $\genericrelation$, let $\genericrelation^{-1} = \{ (y,x) : (x, y) \in~\genericrelation\}$.
For $x \in \genericworld$, let $\genericrelation(x) = \{y \in \genericworld : x\genericrelation y\}$; and for $S \subseteq \genericworld$, let $\genericrelation(S) = \bigcup_{x \in S} \genericrelation(x)$.
We use the standard grammar of modal logic and entailment for a model ($\mentailgeneric{\phi}$).
We denote the Necessitation rule ($\frac{A}{\necc A}$) as (Nec) and we have the modal formula $K := \necc (A \rightarrow B) \rightarrow (\necc A \rightarrow \necc B)$ that holds in all frames.
Other common modal formula are defined along with their corresponding properties of frames:
\begin{equation*}
    \begin{aligned}
        a4 := \poss\poss A \rightarrow \poss A,~\text{transitivity}; 
        && aT := A \rightarrow \poss A,~\text{reflexivity}; \\
        aD := \poss \top,~\text{seriality}; 
        && ad := \poss A \rightarrow \poss\poss A,~\text{density}; \\
        ad_2 := \poss A \land \poss B \rightarrow \poss(\poss A \land \poss B),~\text{2-density}; 
        && a2 := \poss\necc A \rightarrow \necc \poss A,~\text{confluence.\footnotemark}
    \end{aligned}
\end{equation*}
\footnotetext{Confluence is also referred to as the Church-Rosser property, the diamond property, or being weakly directed.}

The smallest modal logic is $\mK$, which contains the standard propositional logic axioms, \emph{modus ponens}, $K$, and (Nec).
A normal modal logic is the smallest modal logic that consists of $\mK$; any extra (valid) formulae $a_1, a_2 \dots$; and any formula that can be obtained from applying the rules in $\mK$ on the formulae given.
The respective logic is denoted $\mK + a_1 + a_2 +\dots$~.

We denote some normal modal logics derived from the formulas given above:
\begin{equation*}
    \begin{aligned}
        & \mathbf{K4} := \mK + a4, & \mathbf{T} := \mK + aT, \\
        & \mathbf{D4} := \mathbf{K4}+ aD, & \mathbf{S4} := \mathbf{K4} + aT, \\
        & \mathbf{OI} := \mathbf{D4} + ad_2, & \bm{\Lambda.2} := \bm{\Lambda} + a2;
    \end{aligned}
\end{equation*}
where $\bm{\Lambda}$ is an arbitrary normal modal logic.
Each of these has a combination of the properties of the corresponding modal formula.

As standard, for a set of formulae, $\Gamma$, and a formula, $\phi$, let $\Gamma \vdash \phi$ mean that $\phi$ can be logically derived from $\Gamma$ ($\phi \in \mK + \Gamma$) and $\Gamma \nvdash \phi$ iff $\Gamma \vdash \lnot \phi$.
For a frame, $F = \genericframe$, we write $\modal{F}$ to mean the set of formula that are satisfied under all possible evaluations of $F$.
This set of formula could be some logic, $\bm{\Lambda}$, and we may write $\modal{F} = \bm{\Lambda}$.


We introduce Gabbay's Irreflexive Rule~\cite{Gabbay1981}.
The rule is
\begin{prooftree}
\AxiomC{$\lnot (p \rightarrow \poss p) \rightarrow \phi$}
\LeftLabel{$\gabb$}
\RightLabel{if $p$ does not occur in $\phi$.}
\UnaryInfC{$\phi$}
\end{prooftree}
For a normal modal logic $\bm{\Lambda}$ without reflexivity ($aT$), let $\bm{\Lambda} + \gabb$ be the smallest modal logic of $\bm{\Lambda}$ with the addition of being closed under $\gabb$.
A frame whose logic is $\bm{\Lambda} + \gabb$ is irreflexive.

\section{Spacetime and Modal Logic}
\label{sec:modalspacetime}
We begin by looking at modal logic within the context of a generic spacetime.
No assumptions are made on where the spacetime lies on the causal ladder.
We analyse the normal modal logics that act as a subset for the modal logic of spacetimes.

Given the definitions of the various causal relations, the relational properties they have (known from Lemma~\ref{lem:relcommonproperties}), and the properties on frames that normal modal logics describe; we have the following results on all spacetimes.
\begin{corollary}
    Let $\genericworld$ be a spacetime. Then, we have that
    \begin{enumerate}[label=(\roman*)]
        \item $\mathbf{S4} \subseteq \modal{\mframe{\genericworld}{\chroneq}}, \modal{\mframe{\genericworld}{\caus}}$;
        \item $\mathbf{OI} \subseteq \modal{\mframe{\genericworld}{\chron}}$;
        \item $\mathbf{D4} + ad \subseteq \modal{\mframe{\genericworld}{\after}}$.
    \end{enumerate}
    \label{cor:spacetimeclass}
\end{corollary}

This can be seen simply from the fact that the normal modal logics yield frames that have the associated properties and the causal relations have those properties.
For instance, $\mathbf{S4}$ yields frames that are transitive and reflexive, and $\caus$ is a transitive and reflexive relation.

\subsection{The After Formula}
\label{sec:afterformula}
In \cite{Shapirovsky2005}, it was shown that the formula, which we refer to as the \emph{after formula},
\begin{equation*}
    \afterformula := \poss ( \poss (p_1 \land \lnot p_2 \land \necc \lnot p_2)
                \land \poss (p_2 \land \lnot p_1 \land \necc \lnot p_1))
                \land \poss q
                \rightarrow
                \poss (\poss p_1 \land \poss q) \lor \poss (\poss p_2 \land \poss q),
\end{equation*}
holds in the modal logic of any $1+n$-dimensional Minkowski spacetime with the $\after$ relation.
This formula corresponds to the first-order property:
\begin{equation}
\label{eq:afterformula:firstorder}
\begin{aligned}
    \forall x, y, y_1, y_2, z
    &\Big[ \left(x \genericrelation y \land y \genericrelation y_1 \land y \genericrelation y_2 \land x \genericrelation z
    \land y_1 \neq y_2 \land \lnot(y_1  \genericrelation  y_2) \land \lnot(y_2  \genericrelation  y_1) \right)
    \\ &\rightarrow
    \exists t \left(x \genericrelation t \land t \genericrelation z \land (t \genericrelation y_1 \lor t \genericrelation y_2 ) \right) \Big].
    \footnotemark
\end{aligned}
\end{equation}
\footnotetext{Note $y_1 \neq y_2 \land \lnot(y_1  \genericrelation  y_2) \land \lnot(y_2  \genericrelation  y_1) \equiv y_1 \genericrelation^{\bowtie} y_2$ in \cite{Shapirovsky2005}.}
Observe that $\afterformula$ is a modified version of 3-density: if two (unrelated) points are related by a common point from the root and there is a third point that can be reached from the root, then at least one of the two points and the third point must share a common point from the root.
In fact, $\afterformula$ is a Sahlqvist formula (see \cite{Blackburn2001}), which provides properties \change{such as canonicity, first-order correspondence (as already demonstrated in Equation~\eqref{eq:afterformula:firstorder}), etc.}

We make the following observation of $\afterformula$ within the scope of normal modal logics:
\begin{lemma}
    We have
    \begin{enumerate*}[label=(\roman*), itemjoin = \quad]
        \item $\mathbf{D4.2} + ad \nvdash \afterformula$; \label{lem:aaf:d4.2d}
        \item $\mathbf{T} \nvdash \afterformula$; \label{lem:aaf:t}
        \item $\mathbf{K4} + ad_2 \vdash \afterformula$. \label{lem:aaf:k4d2}
    \end{enumerate*}
    \label{lem:afterformula}
\end{lemma}
\begin{proof}
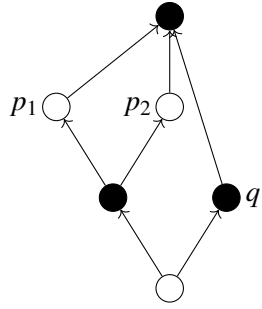
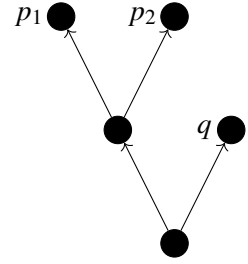
\begin{figure}[t]
    \centering
    \begin{subfigure}{.4\textwidth}
    \centering
    \begin{tikzpicture}[nodes={draw, circle}, ->,yscale=.8]
    \node (s) {} [grow'=up]
        child {node[refl] (0) {}
            child {node (00) [label=left:{$p_1$}] {}}
            child {node (01) [label=left:{$p_2$}] {}
                child {node[refl] (T) {}}
            }
        }
        child {node[refl] (1) [label=right:{$q$}] {}
        };
    \draw (00) -- (T);
    \draw (1) -- (T);
    \end{tikzpicture}
    \caption{A frame that is transitive, serial, dense, and confluent ($\mathbf{D4.2} + ad$).}
    \label{fig:afnotin:d4d}
    \end{subfigure}
    \hfill
    \begin{subfigure}{.4\textwidth}
    \centering
    \begin{tikzpicture}[nodes={draw, circle}, ->,]
    \node[refl] (s) {} [grow'=up]
        child {node[refl] (0) {}
            child {node[refl] (00) [label=left:{$p_1$}] {}}
            child {node[refl] (01) [label=left:{$p_2$}] {}}
        }
        child {node[refl] (1) [label=left:{$q$}] {}};
    \end{tikzpicture}
    \caption{A reflexive frame where the relation is non-transitive ($\mathbf{T}$).}
    \label{fig:afnotin:T}
    \end{subfigure}
    \caption{Counter models to $\afterformula$ in specific normal modal logics.}
    \label{fig:afnotin}
\end{figure}

Claims~\ref{lem:aaf:d4.2d} and \ref{lem:aaf:t} are proven by models that act as counterexamples given in Figures~\ref{fig:afnotin:d4d} and \ref{fig:afnotin:T} respectively.

For~\ref{lem:aaf:k4d2}, let $\mathcal{M} = \genericmodel$ be a $\mathbf{K4} + ad_2$ model and let $x \in \genericworld$.
If we have that
\begin{equation*}
\mentail{\mathcal{M}}{x}{\poss ( \poss (p_1 \land \lnot p_2 \land \necc \lnot p_2) \land \poss (p_2 \land \lnot p_1 \land \necc \lnot p_1)) \land \poss q};
\end{equation*}
then $\mentail{\mathcal{M}}{x}{\poss p_1}$ (by transitivity and properties of $\land$).
Since we have $\mentail{\mathcal{M}}{x}{\poss p_1 \land \poss q}$, then by 2-density we must have that $\mentail{\mathcal{M}}{x}{\poss(\poss p_1 \land \poss q)}$ and therefore $\mentail{\mathcal{M}}{x}{\afterformula}$.
Since $x, \mathcal{M}$ are generic, then $\mathbf{K4} + ad_2 \vdash \afterformula$.
\end{proof}

Note that reflexivity implies 2-density.
Additionally, since $\mathbf{D4.2} + ad \nvdash \afterformula$, then any sub-logics of $\mathbf{D4.2} + ad$ ($\mathbf{D4}$, $\mathbf{K4}$, \dots) do not contain $\afterformula$.
Conversely, any logics that contain $\mathbf{K4} + ad_2$ as a sub-logic, such as $\mathbf{S4}$, contain $\afterformula$.
Therefore, since $\afterformula$ holds in transitive and 2-dense (or reflexive) frames, then it holds in frames that use a spacetime and one of $\chron$, $\chroneq$ or $\caus$.

Let $\afterlogic := \mathbf{D4} + ad + \afterformula$.\footnote{In \cite{Shapirovsky2005}, the logic $\mathbf{L\alpha_0} = \afterlogic\mathbf{.2} = \mathbf{D4.2} + ad + \afterformula$ is used for an analysis on Minkowski spacetime.}
Since $ad$ and $\afterformula$ are Sahlqvist formula, then by the Sahlqvist Theorem~\cite{Blackburn2001} canonicty is achieved for the formulas and $\afterlogic$.
\begin{corollary}
    The logic $\afterlogic$ is canonical.
\end{corollary}

\begin{toappendix}
\subsection{Geodesic Normal Coordinates}
\noindent
In this section we define and collect some technical results from causality theory that are needed to prove the after formula holds in general spacetimes. We refer to \cite[\S 5]{landsman2021FoundationsGeneralRelativity} for details.

Given any subset $A\subseteq M$, we write $x\chron_A y$ if there exists a fd timelike curve from $x$ to $y$ that lies entirely within $A$. Clearly $x\chron_A y$ implies~$x\chron y$. Define~$\caus_A$ and~$\after_A$ similarly.

A curve $\gamma$ is called a \emph{geodesic} if it is length extremising (technically: $\nabla_{\dot\gamma}\dot\gamma=0$ with respect to the Levi-Civita connection), which is interpreted as the motion of a non-accelerating, free falling body. For the present work it is important to note that geodesics are locally uniquely determined by their starting position and initial velocity. This ensures the existence of the \emph{exponential map}: for $x\in M$ it is the smooth map
$\exp_x\colon \mathcal V_x\subseteq T_xM\to M$ defined on a suitable open neighbourhood $\mathcal V_x$ of $0$ by $\exp_x(v):=\gamma_v(1)$, where $\gamma_v$ is the unique geodesic with $\gamma_v(0)=x$ and $\dot\gamma_v(0)=v$. The region $\mathcal{V}_x$ can be shrunk to a convex open neighbourhood of $0$ such that $\exp_x$ becomes a diffeomorphism onto its image $\exp_x(\mathcal{U}_x)=U_x$, which is called a \emph{(convex) normal neighbourhood} of $x$. This diffeomorphism facilitates what is known as \emph{geodesic normal coordinates}. Under these coordinates, geodesics within the neighbourhood $U_x$ emanating from $x$ correspond precisely to straight line segments in $T_xM$ through the origin.

\begin{theorem}\label{thm:GNC}
    Every $x\in M$ admits an open convex normal neighbourhood $U_x\subseteq M$ together with a diffeomorphism~${\Phi:=\exp_x^{-1}\colon U_x\to \mathcal{U}_x}$ onto an open convex neighbourhood $\mathcal{U}_x$ of the origin in Minkowski space,
    such that $\Phi(x)=0$, affinely parametrised geodesics $\gamma$ in $U_x$ with $\gamma(0)=x$ are transformed into straight lines as~$\Phi\circ \gamma (t)= t\dot{\gamma}(0)$, and for every $y\in U_x$ we have:
    \begin{equation}
    \label{eq:local relations}\tag{$\star$}
        x \chron_{U_x} y \iff 0 \chron \Phi(y),
        \qquad
        x \caus_{U_x} y \iff 0 \caus \Phi(y),
        \qquad
        x \after_{U_x} y \iff 0 \after \Phi(y).
    \end{equation}
\end{theorem}
\begin{proof}
    See~\cite[Theorem~5.5]{landsman2021FoundationsGeneralRelativity} and \cite[Corollary~2.10]{minguzzi2019LorentzianCausalityTheory}.
\end{proof}

\begin{remark}
    Of course, the dual of Equation~\eqref{eq:local relations} also holds, where $y$ appears in the past of $x$.
    Note however that $\Phi$ does not preserve the causal relations between arbitrary points in $U_x$: only to and from the basepoint~$x$.
\end{remark}
\end{toappendix}

\subsection{After Formula in Spacetime}
\label{sec:spacetime:after}
As mentioned previously, it was shown in~\cite{Shapirovsky2005} that the after formula is satisfied in the after modality of any Minkowski space.

\begin{proposition}[\cite{Shapirovsky2005}]\label{thm:minkowskiafter}
    Consider $1+n$-dimensional Minkowski space $\mathbb{R}^{1+n}$ with its after relation~$\after$. Then
    \begin{equation*}
        \entail{\mframe{\mathbb{R}^{1+n}}{\after}}{}{a\alpha f}.
    \end{equation*}
\end{proposition}

Now we generalise this result, showing that $\afterformula$ holds in \emph{any} spacetime.
The main technical step towards proving this is to show the following property: if $x\after y \after y_1,y_2$ and $y_1\neq y_2$ are spacelike separated, then $x\chron y_i$ for some~$i$. First we give an elementary proof in Minkowski space (Lemma~\ref{lem:mink:horismos} and Corollary~\ref{corollary:unique lightline in Minkowski}), which is then lifted to arbitrary spacetimes using geodesic normal coordinates~(Theorem~\ref{thm:GNC}). \change{The key intuition is that in the extreme scenario that~$x\after y\after y_j$ represents a lightline, the fact that~$y_i$ is spacelike separated from~$y_j$ means that it must necessarily move away from this lightline, and so must be in the interior cone of~$x$. This behaviour lifts from Minkowski space to general spacetimes.}

\begin{toappendix}
\subsection{Proof of Theorem~\ref{thm:spacetimeafter}}
As mentioned in~Section~\ref{sec:spacetime:after}, the main technical step is to prove: if $x\after y \after y_1,y_2$ and $y_1\neq y_2$ are spacelike separated, then $x\chron y_i$ for some~$i$. We first do this in Minkowski space.

\begin{lemma}
    In Minkowski space, let $x\horismos y \horismos y_1,y_2$ and $x\horismos y_1,y_2$. Then $y_1\horismos y_2$ or $y_2\horismos y_1$.
    \label{lem:mink:horismos}
\end{lemma}
\begin{proof}
    Take some arbitrary $z$ in the intersection of the lightcones of $x$ and $y$. Then we get from Equation~\eqref{eq:minkowski causal relations} that $\norm{\vec{z}-\vec{x}}= z^0- x^0$ and $\norm{\vec{z}-\vec{y}}=z^0-y^0$, and hence
    $$
    \norm{\vec{z}-\vec{x}}-\norm{\vec{z}-\vec{y}} = y^0 -x^0 = \norm{\vec{y}-\vec{x}}.
    $$
    This shows that equality is attained in the triangle inequality:
    $$
    \norm{\vec{z}-\vec{x}}
    \leq 
    \norm{\vec{z}-\vec{y}} + \norm{\vec{y}-\vec{x}},
    $$
    which occurs precisely when there exists a scalar $\lambda\geq 0$ with $\vec{z}-\vec{y}=\lambda(\vec{y}-\vec{x})$.
    In particular, for our situation, we get $\lambda_1,\lambda_2\geq 0$ with $\vec{y_i} = \vec{y}+ \lambda_i(\vec{y}-\vec{x})$. Note now by $y\horismos y_i$ that
    $$
    y_i^0 = y^0 + \norm{\vec{y_i}-\vec{y}}
    =
    y^0 + \lambda_i\norm{\vec{y}-\vec{x}} = y^0 + \lambda_i(y^0-x^0),
    $$
    so in particular $y_2^0 - y_1^0 = (\lambda_2-\lambda_1)(y^0-x^0)$. Similarly note $\vec{y_2}-\vec{y_1} = (\lambda_2-\lambda_1)(\vec{y}-\vec{x})$, and so we get
    $$
    \norm{\vec{y_2}-\vec{y_1}} = |\lambda_2-\lambda_1|\norm{\vec{y}-\vec{x}}
    =
    |\lambda_2-\lambda_1|(y^0-x^0)
    = 
    |y_2^0 - y_1^0|,
    $$
    which is precisely what it means for $y_1\horismos y_2$ or $y_2\horismos y_1$ to hold.
\end{proof}

\begin{corollary}\label{corollary:unique lightline in Minkowski}
    In Minkowski space, if $x\after y\after y_1,y_2$ and $y_1\neq y_2$ are spacelike separated, then~${x\chron y_i}$ for some $i$.
\end{corollary}
\begin{proof}
    Transitivity implies $x\after y_1,y_2$. So either $x\chron y_i$ for some $i$, in which case we are done, or $x\horismos y_1,y_2$. Moreover, if $x\chron y$ or $y\chron y_i$ for some $i$ we get the desired relation via the push-up rule. We are therefore left to consider the case that $x\horismos y\horismos y_1,y_2$. However, by Lemma~\ref{lem:mink:horismos}, this would imply $y_1\horismos y_2$ or $y_2\horismos y_1$, contradicting that $y_1$ and $y_2$ are spacelike separated.
\end{proof}

The next few lemmas generalise this result to arbitrary spacetimes. We need some more definitions. A subset $A\subseteq M$ is called \emph{achronal} if $I^{+}(A)\cap A=\varnothing$, and a curve $\gamma$ is called \emph{achronal} if $\im(\gamma)$ is achronal, \change{where we recall for~$\gamma\colon I\to M$ that~$\im(\gamma)=\{\gamma(t):t\in I\}$.} Explicitly, this means there are no times $s,t$ so that~$\gamma(s)\chron \gamma(t)$. A geodesic from~$p$ to~$q$ is called \emph{maximising} if it has greater or equal length than every other causal curve from~$p$ to~$q$. We then have the powerful result~\cite[Theorem~2.22]{minguzzi2019LorentzianCausalityTheory}, saying achronal lightlike geodesics are precisely the maximising lightlike ones.

\begin{theorem}
\label{thm:2.22}
    Let $\gamma$ be a causal curve connecting $p$ to $q\neq p$. Either: there is a timelike curve~$\sigma$ from $p$ to $q$ whose length is strictly greater than that of $\gamma$; or $\gamma$ is a maximising geodesic (up to parametrisation).
    In particular, if there are no timelike curves connecting $p$ to $q$ then $\gamma$ is an achronal lightlike geodesic (up to parametrisation).    
\end{theorem}

\begin{lemma}\label{lem:corner rounding}
Let $U\subseteq M$ be open, and take $p,q\in U$. If there exists a causal curve~$\gamma$ in $U$ from~$p$ to~$q\neq p$ that is not an achronal lightlike geodesic, then $p\chron_U q$.
\end{lemma}
\begin{proof}
Consider the spacetime $(U,g|_U)$. Applying Theorem~\ref{thm:2.22} to the given curve $\gamma$ in $U$ from~$p$ to~$q\neq p$, either we get a timelike curve $\sigma$ facilitating $p\chron_U q$, in which case we are done, or $\gamma$ is a maximising geodesic. If $\gamma$ is itself timelike we are also done, but if $\gamma$ is a maximising lightlike geodesic then it is achronal, contradicting the hypothesis.
\end{proof}

\begin{lemma}\label{lem:lift chronology}
Let $x\in M$, and take geodesic normal coordinates $\Phi\colon U_x\to \mathcal{U}_x$ from Theorem~\ref{thm:GNC}.
Suppose $p,q\in U_x$ satisfy~${p\after_{U_x} x \after_{U_x} q}$. If $\Phi(p)\chron \Phi(q)$ in Minkowski space, then $p\chron_{U_x} q$.
\end{lemma}
\begin{proof}
If $p\chron_{U_x} x$ or $x\chron_{U_x} q$ then the desired result follows from the push-up rule. Suppose therefore that $p\horismos_{U_x} x$ and $x\horismos_{U_x} q$. Let $\alpha$ and $\beta$ be the curves facilitating these relations, and let $\gamma$ be their concatenation. We claim that $\gamma$ is not a lightlike geodesic, so in particular not an achronal lightlike geodesic. From that, the desired~$p\chron_{U_x}q$ will follow by~Lemma~\ref{lem:corner rounding}.

Suppose for the sake of contradiction that $\gamma$ is a lightlike geodesic, affinely parametrised and so that~${\gamma(0)=x}$. Then the curve $\Phi\circ \gamma$ is the straight line passing through the origin with velocity~$\dot{\gamma}(0)$, which immediately gives~$\Phi(p)\horismos \Phi(q)$, a contradiction to~$\Phi(p)\chron \Phi(q)$.
\end{proof}
\end{toappendix}

\begin{lemmarep}\label{lemma:unique lightlines in any spacetime}
    In any spacetime, if $x\after y \after y_1,y_2$ and $y_1\neq y_2$ are spacelike separated, then $x\chron y_i$ for some~$i$.
\end{lemmarep}
\begin{toappendix}
  \label{apx:lemma:unique lightlines in any spacetime}
\end{toappendix}
\begin{proof}
    Let $\gamma,\delta_1,\delta_2$ denote the causal curves corresponding to $x\after y$ and $y\after y_1,y_2$, respectively. Without loss of generality, suppose that $\delta_i$ are parametrised on the unit interval $[0,1]\to M$.
    Consider the set~${S= \delta_1^{-1}(\im(\delta_2))}$, which is nonempty since $\delta_1(0)=\delta_2(0)$ and a strict subset of $[0,1]$ because~${y_1\neq y_2}$ are spacelike separated.
    Since~$[0,1]$ is compact and $\delta_2$ is continuous,~$\im(\delta_2)$ is compact in~$M$. Since~$M$ is Hausdorff,~$\im(\delta_2)$ is thus closed, and by continuity of $\delta_1$ the preimage~$S$ is closed. Hence $S$ has a maximum element $t_0$, the latest time at which~$\delta_1$ overlaps~$\delta_2$. We can further reparametrise $\delta_2$ without loss of generality such that $\delta_1(t_0)=\delta_2(t_0)$, and hence $\delta_1$ and $\delta_2$ never intersect after~$t_0$. Let~${y':= \delta_1(t_0)=\delta_2(t_0)}$ be the last intersection point; see the left diagram in Figure~\ref{fig:uniquelightlines}.
    Note if~$t_0=0$ then~$y'= y$.
    
    Using Theorem~\ref{thm:GNC}, choose a convex normal neighbourhood $U_{y'}\subseteq M$ and geodesic normal coordinates
    $\Phi\colon U_{y'}\to \mathcal{U}_{y'}$.
    Since $U_{y'}$ is open and $\delta_i$ are continuous with $\delta_i(t_0)=y'$, we can find time intervals~$\epsilon_i>0$ such that~$t_0+\epsilon_i\leq 1$ and on $[t_0,t_0+\epsilon_i]$ the curves~$\delta_i$ stay in~$U_{y'}$.
    For any~${t\in (t_0,t_0+\epsilon_i]}$ we then get that $y'\after_{U_{y'}} \delta_i(t)$, and by~\eqref{eq:local relations} we have $0\after \Phi\circ \delta_i(t)$ in $\mathcal{U}_{y'}$.
    From this it follows that~${\delta_i(t) \neq y'}$ for every~$t\in (t_0,t_0+\epsilon_i]$, otherwise~$y'\after_{U_{y'}} y'$ and hence $0\after 0$, which is impossible in Minkowski space.
    Using this, we see from Equation~\eqref{eq:minkowski causal relations} that the time coordinate~${(\Phi\circ \delta_i(t))^0}$ must be strictly greater than zero for all~${t\in (t_0,t_0+\epsilon_i]}$.
    In particular there exists some~$\epsilon >0$ that is below the endpoints:~${\epsilon < (\Phi\circ \delta_i(t_0+\epsilon_i))^0}$, meaning that eventually the curves $\Phi\circ \delta_i$ are above the hyperplane~${\{ x\in \mathbb{R}^{1+n}: x^0= \epsilon\}}$.
    By the intermediate value theorem there exists~${t_i\in (t_0,t_0+\epsilon_i]}$ so that the curves are precisely on the hyperplane:~$(\Phi\circ\delta_i(t_i))^0 = \epsilon$.
    Using this, define the distinct points~$p_i:=\delta_i(t_i)$, from which we get the distinct points~$z_i:=\Phi(p_i)$ on the hypersurface, meaning that $z_1\neq z_2$ are spacelike separated.
    
    We similarly pick a point $p$ in the past of $y'$. If $t_0>0$, find some~$\epsilon_0>0$ so that $\delta_1$ remains within~$U_{y'}$ on the interval $[t_0-\epsilon_0,t_0]$, and define $p:= \delta_1(t_0-\epsilon_0)$ and~$z:=\Phi(p)$; see the right diagram in Figure~\ref{fig:uniquelightlines}. In the case that $t_0=0$, so $y'=y$, pick $p$ analogously but on the curve $\gamma$ instead.

    In total, we have the situation in the normal neighbourhood $U_{y'}$ that $p\after_{U_{y'}} y' \after_{U_{y'}} p_1,p_2$. Using~\eqref{eq:local relations} and its dual, we get the situation in the region $\mathcal{U}_{y'}$ of Minkowski space that~$z\after 0 \after z_1,z_2$ where~$z_1\neq z_2$ are spacelike separated, so by~Corollary~\ref{corollary:unique lightline in Minkowski} we get~$z\chron z_k$ for some~$k$.
    Applying Lemma~\ref{lem:lift chronology} we hence get $p\chron_{U_{y'}} p_k$.

    Finally, by construction we have $x\after p$ and $p_i\after y_i$, so globally $x\after p \chron p_k \after y_k$, and the desired result follows by the push-up rule:~$x\chron y_k$.
\end{proof}

\change{Note that the proof (given in Appendix~\ref{apx:lemma:unique lightlines in any spacetime}) still works if~$x=y$, in which case from~$x\after x$ we get a causal curve from~$x$ to itself, which is all that is needed for the proof. A visualisation of the proof of Lemma~\ref{lemma:unique lightlines in any spacetime} is in Figure~\ref{fig:uniquelightlines}.}

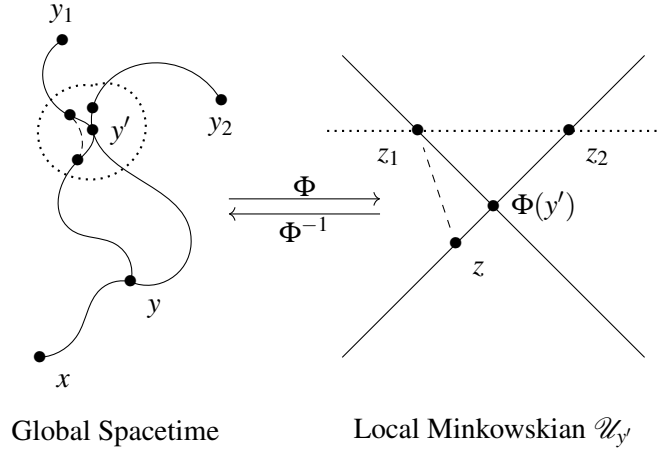
\begin{figure}[t]
    \centering
    \begin{tikzpicture}
    \node[] (gen) at (1,-1) {Global Spacetime};
    \node[label=below right:{$x$}] (x) {$\bullet$};
    \node[label=below right:{$y$}] (y) at (1.2,1) {$\bullet$};
    \draw [] (x.center) to [ curve through ={(.5,.3) . . (.8,.9)  }] (y.center);
    
    \node[label=right:{$y'$}] (yp) at (.7,3) {$\bullet$};
    \node[label=above:{$y_1$}] (y1) at (.3,4.2) {$\bullet$};
    \node[label=below:{$y_2$}] (y2) at (2.4,3.4) {$\bullet$};
    \node[] (diz) at (.5,2.6) {$\bullet$}; 
    \node[] (diz1) at (.4,3.2) {$\bullet$}; 
    \node[] (diz2) at (.7, 3.3) {$\bullet$}; 
    \draw [] (y.center) to [ curve through ={(1,1.5) . . (.4,1.7)  . . (diz) . . (yp) . . (diz1) . . (.1,4)  }] (y1.center);
    \draw [] (y.center) to [ curve through ={(2,1.5)  . . (yp) . .  (diz2) . . (1.6,3.9) }] (y2.center);
    \draw [densely dashed] (diz.center) to [ curve through ={(.55,2.95)}] (diz1.center);
    \node[] (U) at (.7,3.6) {};
    \draw [dotted,thick] (U.center) to [ curve through ={(0,2.8)  . . (.4,2.4) . . (1.3,3.4)  }] (U.center);

    \node[label=above:{$\Phi$},label=below:{$\Phi^{-1}$}] at (3.5,2) {};
    \draw[->] (2.5,2.1) -- (4.5,2.1);
    \draw[->] (4.5,1.9) -- (2.5,1.9);

    \node[] (mink) at (6, -1) {Local Minkowskian $\mathcal{U}_{y'}$};
    \node[label=right:{$\Phi(y')$}] (dyp) at (6, 2) {$\bullet$};
    \draw[] (4,0) -- (dyp.center) -- (4, 4);
    \draw[] (8,0) -- (dyp.center) -- (8, 4);
    \node[label=below right:{$z$}] (z) at (5.5,1.5) {$\bullet$};
    \node[label=below left:{$z_1$}] (z1) at (5,3) {$\bullet$};
    \node[label=below right:{$z_2$}] (z2) at (7,3) {$\bullet$};
    \draw[dashed] (z) -- (z1);
    \draw[thick, dotted] (3.8,3) -- (8.2,3);
\end{tikzpicture}
    \caption{Visualisation of the proof for Lemma~\ref{lemma:unique lightlines in any spacetime} using two-dimensional Minkowski spacetime.
    Dashed (solid) lines represent chronological~(causal)~curves. The dotted region on the left denotes~$U_{y'}$, and the dotted line on the right depicts the spacelike hypersurface defined by $x^0=\epsilon$.}
    \label{fig:uniquelightlines}
\end{figure}

\begin{theorem}\label{thm:spacetimeafter}
    Let $\genericworld$ be a spacetime and $\after$ its after relation. Then
    \begin{equation*}
        \entail{\mframe{\genericworld}{\after}}{}{a\alpha f}.
    \end{equation*}
\end{theorem}
\begin{proof}
    Take $x, y, y_1, y_2, z \in M$ such that $x \after y, y_1, y_2, z$ and $y \after y_1, y_2$, where $y_1 \neq y_2$ are unrelated by~$\after$. This implies $y_1$ and $y_2$ are spacelike separated. We need to find $t\in\genericworld$ such that $x\after t \after y_i, z$, for some~$i$. First, by Lemma~\ref{lemma:unique lightlines in any spacetime} we can assume without loss of generality that $x\chron y_1$. In the case~$x=z$ the result is trivial~(pick $t=x$), so assume~$x\neq z$. Thus we have a case distinction of whether $x\chron z$ or $x\horismos z$ (Lemma~\ref{lemma:relations between causal orders}). In the first case, the existence of $t\in \genericworld \setminus \{x\}$ with $x\chron t\chron y_1,z$ follows immediately from 2-density of $\chron$.
    
    Instead, suppose that $x \horismos z$. Let~$\gamma$ be the timelike curve connecting $x\chron y_1$ and let~$\delta$ be the lightcurve connecting $x\horismos z$. Take a point $p\in \im(\gamma)\setminus\{x\}$ so that $x\chron p\chron y_1$.
    Then $x\in I^{-}(p)$, and since chronological cones are open there must exist a point $t\in I^{-}(p)\cap \im(\delta)\setminus\{x\}$, which hence satisfies $x\after t \after y_1$ and $x \after t \after z$, as desired.
\end{proof}

Finally, we build on Corollary~\ref{cor:spacetimeclass} using Theorem~\ref{thm:spacetimeafter} and Lemma~\ref{lem:afterformula}.
\begin{corollary}
    For any spacetime $\genericworld$ we have that $\afterlogic \subseteq \modal{\mframe{\genericworld}{\after}}$.
\end{corollary}

What this section shows is that there is a connection between the push-up rule, which naturally occurs in spacetimes, and $\afterformula$.
Many of the proofs that we have used for $\afterformula$ are reliant on the push-up rule.
An exploration of the push-up rule and its link to Kripke frames (possibly with multiple relations) is warranted, but left to future work.

\section{Two-dimensional Spacetimes}
\label{sec:separable}
It remains an open problem as to what the modal logic of Minkowski spacetime with $\after$ is, but it is known that the logic of two-dimensional Minkowski spacetime is unique to any $1+n$-dimensional Minkowski spacetime ($n > 2$)~\cite{Shapirovsky2005}.

\subsection{After Formula for Two Dimensions}
\label{sec:separable:after2}
Consider the formula,
\begin{equation*}
    \aftertwoformula := \poss (p_1 \land \lnot p_2 \land \necc \lnot p_2)
                \land \poss (p_2 \land \lnot p_1 \land \necc \lnot p_1)
                \land \poss q
                \rightarrow
                \poss (\poss p_1 \land \poss q) \lor \poss (\poss p_2 \land \poss q),
\end{equation*}
which is $\afterformula$ without the first diamond.\footnote{Read as the after-2 formula/axiom.}
It corresponds to the following first order property:
\begin{equation*}
\begin{aligned}
    \forall x, y_1, y_2, z & \Big[ [x \genericrelation y_1 \land x \genericrelation y_2 \land x \genericrelation z 
    \land y_1 \neq y_2 \land \lnot(y_1  \genericrelation  y_2) \land \lnot(y_2  \genericrelation  y_1) ]
    \\ &  \quad \to
    \exists t (x \genericrelation t \land t \genericrelation z \land (t \genericrelation y_1 \lor t \genericrelation y_2 )) \Big].
\end{aligned}
\end{equation*}
This formula is a Sahlqvist formula and, therefore, any canonical normal modal logic that adds $\aftertwoformula$ is canonical.
Even though $\aftertwoformula$ is a simple modification to $\afterformula$, we can show that it is strictly more expressive within the context of spacetimes.

\begin{figure}[t]
    \centering
        \centering
        \begin{tikzpicture}[nodes={draw, circle}, ->,]
        \node (s) {} [grow'=up]
        child {node[refl] (0) [label=left:{$p_1$}] {}}
        child {node[refl] (1) [label=left:{$p_2$}] {} }
        child {node[refl] (2) [label=left:{$q$}] {} };
        \end{tikzpicture}
        \caption{A finite $\afterlogic$ model that does not have $\aftertwoformula$ on the root.\label{fig:aftertwoexpress}}
\end{figure}
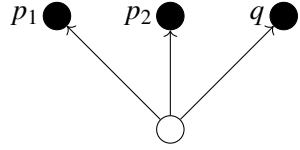

\begin{theorem}
We have that
    \begin{enumerate*}[label=(\roman*)]
        \item $\mathbf{K4} + \aftertwoformula \vdash \afterformula$; \label{thm:aftertwoexpress:expressive}
        \item $\afterlogic \nvdash \aftertwoformula$. \label{thm:aftertwoexpress:inexpressive}
    \end{enumerate*}
    \label{thm:aftertwoexpress}
\end{theorem}
\begin{proof}
    \ref{thm:aftertwoexpress:expressive}:
    Let $\mathcal{M} = \mmodel{\genericworld}{\genericrelation}{\genericeval}$ be a $\mathbf{K4} + \aftertwoformula$ model and $x \in \genericworld$.
    Assume that $\mentail{\mathcal{M}}{x}{\poss(\poss (p_1 \land \lnot p_2 \land \lnot \necc p_2) \land \poss(p_2 \land \lnot p_1 \land \lnot \necc p_1)) \land \poss q}$.
    Let $x \genericrelation y$ such that $\mentail{\mathcal{M}}{y}{\poss (p_1 \land \lnot p_2 \land \lnot \necc p_2) \land \poss(p_2 \land \lnot p_1 \land \lnot \necc p_1)}$, then $\mentail{\mathcal{M}}{x}{\poss \poss (p_1 \land \lnot p_2 \land \lnot \necc p_2) \land \poss\poss(p_2 \land \lnot p_1 \land \lnot \necc p_1)}$.
    
    Since $\mathcal{M}$ is transitive, then we have
    $\mentail{\mathcal{M}}{x}{\poss (p_1 \land \lnot p_2 \land \lnot \necc p_2) \land \poss(p_2 \land \lnot p_1 \land \lnot \necc p_1) \land \poss q}$.
    As $\mathcal{M}$ is a $\aftertwoformula$ model, we have $\mentail{\mathcal{M}}{x}{\poss (\poss p_1 \land \poss q) \lor \poss (\poss p_2 \land \poss q)}$ and therefore $\mentail{\mathcal{M}}{x}{\afterformula}$.
    Because $x$ and $\mathcal{M}$ are general, then $\mathbf{K4} + \aftertwoformula \vdash \afterformula$.

    \ref{thm:aftertwoexpress:inexpressive}: A counterexample is provided by the finite model given in Figure~\ref{fig:aftertwoexpress}.
    The frame has all the properties of $\afterlogic$, but cannot express $\aftertwoformula$.
    One of the points $p_1$ or $p_2$ must share a common point with $q$ after the root but they do not.
\end{proof}

\begin{remark}
For finite (spacetime) frames, $\aftertwoformula$ says that if there are successor clusters from an irreflexive point, there are at most two of them each possibly with their own unique chains of clusters (light lines).
\end{remark}

\subsection{Separability of Two-Dimensional Spacetimes}
\label{sec:separable:separability}
With the after relation, we can see that $\aftertwoformula$ holds in two-dimensional Minkowski spacetime but does not hold in any other $1+n$-dimensional Minkowski spacetime.

\begin{lemma}
    \label{lemma:aftertwoinminkowski}
    Let $M = \mathbb{R}^{1+n}$ and $\after$ be the after relation on Minkowski spacetime.
    We have that
    \begin{enumerate}[label=(\roman*)]
        \item $\mentail{\mframe{\mathbb{R}^2}{\after}}{}{\aftertwoformula}$
        \item $\nmentail{\mframe{\mathbb{R}^{1+n}}{\after}}{}{\aftertwoformula}$ for $n \geq 2$
    \end{enumerate}
\end{lemma}
\begin{proof}
The argument is similar to~\cite[Lemma 6.2]{Shapirovsky2005}.
\end{proof}

It turns out that $\aftertwoformula$ holds in any two-dimensional spacetime.

\begin{theoremrep}
Let $M$ be a two-dimensional spacetime.
Then, 
\begin{equation*}
\mentail{\mframe{M}{\after}}{}{\aftertwoformula}.    
\end{equation*}
\end{theoremrep}
\begin{proof}
Let $x \after y_1, y_2, z$ with $y_1 \neq y_2$ and $y_i \not\after y_j$.
We need to prove there exists $t\in M$ such that $x\after t\after z, y_i$ for some $i$.
Again, in the case that $x=y_i$ or $x=z$ simply take $t=x$. Thus assume all points are distinct from~$x$, so $\after$ is described via Lemma~\ref{lemma:relations between causal orders}.

If $x \chron y_i$ for some $i$ and $x \chron z$, then by 2-density of $\chron$ we are done.

Suppose now that $x \horismos y_1, y_2, z$, and let $\delta_1,\delta_2,\gamma$ be the corresponding causal curves.
Since $x\horismos y_i$, there are no timelike curves from $x$ to $y_i$, and hence by Theorem~\ref{thm:2.22} each $\delta_i$ is an achronal lightlike geodesic (up to reparametrisation).
Similarly, since $x\horismos z$, the curve $\gamma$ is an achronal lightlike geodesic from $x$ to $z$.

Choose a convex normal neighbourhood $U_x$ of $x$, and consider the restrictions of $\delta_1,\delta_2,\gamma$ to~$U_x$. In a two-dimensional spacetime there are exactly two future-directed lightlike velocities at $x$, and hence exactly two future-directed lightlike geodesics emanating from $x$ in $U_x$.
Moreover, since $y_1$ and $y_2$ are spacelike separated, the geodesics $\delta_1$ and $\delta_2$ cannot coincide, and so the velocities of $\delta_1$ and $\delta_2$ at~$x$ correspond to the two distinct lightlike directions in $T_xM$.
Since $\gamma$ is also lightlike, its velocity at~$x$ must therefore be equal to that of, say, $\delta_1$.
Pick a point $t\in \im(\gamma)\cap \im(\delta_1)\cap U_x\setminus\{x\}$ sufficiently close to $x$ so that $t$ lies on the subcurves from $x$ to $z$ and from $x$ to $y_1$, so that $x\after t\after y_1,z$, as desired.

Finally, consider the scenario $x \horismos y_1, y_2$ and $x \chron z$.
Let $\delta_1,\delta_2$ be the achronal lightlike geodesics from $x$ to $y_1,y_2$ (as above), and let $\gamma$ be a timelike curve from $x$ to $z$.
Pick a point ${z'\in \im(\gamma)\setminus\{x,z\}}$ sufficiently close to $x$ so that $x\chron z'\chron z$.
Since $I^{-}(z')$ is open and contains $x$, and $\delta_1$ is continuous with $\delta_1(0)=x$, there exists ${t\in \im(\delta_1)\setminus\{x\}}$ sufficiently close to $x$ such that $t\in I^{-}(z')$.
Then ${x\after t\after y_1}$. Moreover $t\chron z'$ implies $t\after z'\after z$, hence $t\after z$.
Thus $x\after t\after y_1,z$, completing the proof.
\end{proof}

Denote $\aftertwologic := \mathbf{D4} + ad + \aftertwoformula$.
It is clear from Theorem~\ref{thm:aftertwoexpress} that $\afterlogic \subset \aftertwologic$.
\begin{corollary}
    Let $M$ be a $2$-dimensional spacetime. Then $\aftertwologic \subseteq \modal{\mframe{M}{\after}}$.
\end{corollary}

In Lemma~\ref{lemma:aftertwoinminkowski}, we were shown how $\aftertwoformula$ only holds in two dimensions, providing a separation from $1+n$-dimensional Minkowski spacetime.
However, to show a full separation of the logics of 2-dimensional spacetimes from higher dimensional ones, it needs to be shown that $\aftertwoformula$ does not hold in higher dimensional spacetimes.
We note that non-cNTV spacetimes are reflexive in $\after$, meaning that $\aftertwoformula$ holds in those spacetimes.
However, for cNTV spacetimes of dimension greater than 3, it is likely that $\aftertwoformula$ does not hold.
Proving this requires that any high dimensional, cNTV spacetime has at least 3 different light lines from a point irreflexive in $\after$.

We end this section by forming the following conjecture.
\begin{conjecture}
    Let $n \geq 3$.
    For any cNTV $n$-dimensional spacetime, $M$, we have that
    $\nmentail{\mframe{M}{\after}}{}{\aftertwoformula}$ ($\aftertwologic \nsubseteq \modal{\mframe{M}{\after}}$).
\end{conjecture}

\section{Effects of Causal Ladder Properties on Modal Logics}
\label{sec:modalcausalladder}
We now analyse the modal logic of spacetimes on different levels of the causal ladder.
The proofs are fairly obvious throughout, simply relying on the properties of the relevant causal relation.

We have the following obvious result.
\begin{corollary}
    Let $\genericworld$ be a spacetime and $\chron$ and $\after$ be the respective causal relations, then:
    \begin{enumerate}[label=(\roman*)]
        \item if $\chron$ is reflexive, $\mathbf{S4} \subseteq \modal{\mframe{\genericworld}{\chron}}$;
        \item if $\after$ is reflexive, $\mathbf{S4} \subseteq \modal{\mframe{\genericworld}{\after}}$.
    \end{enumerate}
\end{corollary}

\subsection{Totally and Non-totally Vicious}
Totally vicious spacetimes are reflexive in $\chron$ and, consequently, in $\after$.
Note that $\chron$ and $\chroneq$ become indistinguishable from one another ($x \chron y$ iff $x \chroneq y$), and similarly for $\after$ and $\caus$.
What becomes more apparent is that $\chron$ and $\caus$ are indistinguishable.
\begin{lemma}
    Let $\genericworld$ be a totally vicious spacetime and $x, y \in \genericworld$.
    Then $x \chron y$ iff $x \caus y$.
\end{lemma}
\begin{proof}
    If $x \chron y$, there is a timelike path from $x$ to $y$, which is also a causal path.
    Therefore, $x \caus y$.
    If $x \caus y$, then we have that $x \chron y$ since $y \chron y$ and the push-up rule is used to obtain the result.
\end{proof}
This leads to the following result.
\begin{theorem}
    Let $\genericworld$ be a totally vicious spacetime.
    Then,
    \begin{equation*}
        \mathbf{S4} \subseteq \modal{\mframe{\genericworld}{\chron}} = \modal{\mframe{\genericworld}{\chroneq}} = \modal{\mframe{\genericworld}{\after}} = \modal{\mframe{\genericworld}{\caus}}.
    \end{equation*}
\end{theorem}
For non-totally vicious spacetimes, the logics for $\chron$ are restricted since reflexivity is lost.
\begin{theorem}
    If $\genericworld$ is a non-totally vicious spacetime, then $\mathbf{S4} \nsubseteq \modal{\mframe{\genericworld}{\chron}}$.
\end{theorem}

Note that $\after$ and $\caus$ are not affected by a spacetime being NTV ($\after$ is still reflexive), and therefore their modal logics are not affected either.
\begin{corollary}
    If $M$ is NTV and not cNTV, then $\mathbf{S4} \subseteq \modal{\mframe{\genericworld}{\after}} = \modal{\mframe{\genericworld}{\caus}}$.
\end{corollary}

\subsection{Chronological}
Irreflexivity of $\chron$ means that the normal modal logic does not fundamentally change from a non-totally vicious spacetime.
However, by considering $\gabb$, we obtain the following.

\begin{theorem}
    If $M$ is a chronological spacetime, then $\mathbf{OI} + \gabb \subseteq \modal{\mframe{\genericworld}{\chron}}$.
\end{theorem}

Note how the logic of $\after$ or the reflexive causal relations have not changed at the levels of non-totally vicious and chronological from a totally vicious spacetime.

\subsection{Causally Non-Totally Vicious (cNTV)}
Having introduced the notion of cNTV in Section~\ref{sec:spacetime:cntv}, we now see the results of separating the causal condition in the modal logic of $\after$.
\begin{theorem}
    If $\genericworld$ is a causally non-totally vicious spacetime, then $\mathbf{S4} \nsubseteq \modal{\mframe{\genericworld}{\after}}$.
\end{theorem}

When considering a spacetime that is both chronological and cNTV, the two respective theorems are used together to reveal that both $\chron$ and $\after$ have different modal logics from a NTV spacetime.

\subsection{Causal}
A causal spacetime, in a similar way to a chronological spacetime, does not change the normal modal logics of its causal relations, but the logic under $\after$ can be changed if $\gabb$ is considered.
\begin{theorem}
    If $\genericworld$ is a causal spacetime, then $\afterlogic + \gabb \subseteq \modal{\mframe{\genericworld}{\after}}$.
\end{theorem}

\subsection{Distinguishing}
We begin first by extending the notion of a spacetime being distinguishing into the Kripke frame setting.

\begin{definition}
    Let $F = \mframe{\genericworld}{\genericrelation}$.
    We say $F$ is future-distinguishing if for all $x,y \in \genericworld$ whenever ${\genericrelation}(x) = {\genericrelation}(y)$, then $x = y$.
    We say $F$ is past-distinguishing if for all $x,y \in \genericworld$ whenever ${\genericrelation^{-1}}(x) =~{\genericrelation^{-1}}(y)$, then $x = y$.
    Finally, $F$ is distinguishing if it is both past- and future-distinguishing.
\end{definition}

Thus, the spacetime definition of distinguishing is preserved when considering the frame of a spacetime with $\chron$ (since $I^{\pm}(x) = {\genericrelation}^{\pm 1}(x)$).
Our goal is to find some normal modal logic that can capture distinguishing spacetimes.
However, despite being able to maintain definitions across domains, it becomes clear that a past-distinguishing frame has no associated modal formula. \change{The proofs of the following results are by bisimulation, see~\cite{Blackburn2001} for more information.}

\begin{figure}
    \centering
    \begin{tikzpicture}[nodes={draw, circle}, <-,]
    \node[label={[label distance=-.5cm]above:$F_1 = \mframe{W_1}{R_1}$},label=left:{$x$}] (s) {} [grow'=down]
        child {node[refl] (0) [label=left:{$z$}] {}}
        child {node[refl] (1) [label=left:{$y$}] {}};
    \node[label={[label distance=-.5cm]above:$F_2 = \mframe{W_2}{R_2}$},label=right:{$x'$}] (n) at (4,0) {} [grow'=down]
            child {node[refl] (n0) [label=right:{$y'$}] {}};
            
    \draw[->,dotted,thick] (s) to[bend left] (n);
    \draw[->,dotted,thick] (0) to[bend right] (n0);
    \draw[->,dotted,thick] (1) to[bend right] (n0);
    \end{tikzpicture}
    \caption{A bisimulation from a past-distinguishing frame to one that is not.\label{fig:pastdistframes}}
\end{figure}
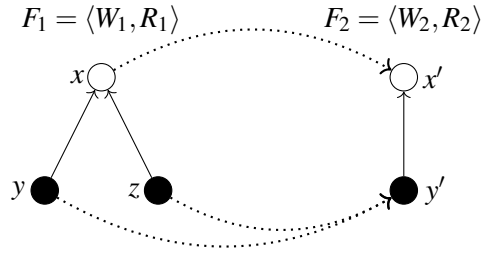

\begin{theorem}
There is no modal formula, $\phi$, such that if $\mframe{\genericworld}{\genericrelation} \models \phi$ then $\mframe{\genericworld}{\genericrelation}$ is past-distinguishing.
\end{theorem}
\begin{proof}
The frames in Figure~\ref{fig:pastdistframes} are bisimilar and so act as a counterexample.
\end{proof}

Similarly, future-distinguishing frames have no associated modal formula.
\begin{theorem}
    There is no modal formula, $\phi$, such that if $\mframe{\genericworld}{\genericrelation} \models \phi$ then $\mframe{\genericworld}{\genericrelation}$ is future-distinguishing.
\end{theorem}
\begin{proof}
A contradiction is obtained by reversing the relation in Figure~\ref{fig:pastdistframes} ($F_1' = \mframe{W_1}{R_1^{-1}}$, $F_2' = \mframe{W_2}{R_2^{-1}}$) and observing a similar bisimulation from $F_1'$ to $F_2'$.
\end{proof}

\begin{corollary}
    A frame being distinguishing is not modally definable.
\end{corollary}
\begin{proof}
    In Figure~\ref{fig:pastdistframes} we have that the frame $F_1$ is (past- and future-)distinguishing, but it is bisimilar to $F_2$, which is not (past-)distinguishing.
\end{proof}

\section{Discussion}
\paragraph{Modal logics of the causal ladder.}
\begin{table}[t]
    \centering
    \caption{Modal logics along the causal ladder.
    Blank cells use the same logic as the one above.}
    \begin{tabular}{|c|c|c|c|}
         \cline{2-4}
         \multicolumn{1}{}{} & \multicolumn{3}{|c|}{Spacetime Modal Logic Containment ($L = \modal{\mframe{M}{\genericrelation}}$)} \\ \hline
         Spacetime Type ($M$) & $\chron$& $\after$ & $\chroneq~/\caus$ \\ \hline
         Totally Vicious & $\mathbf{S4} \subseteq L$ & $\mathbf{S4} \subseteq L$ & $\mathbf{S4} \subseteq L$ \\ \hline
         Non-Totally Vicious & $\mathbf{OI} \subseteq L$ & & \\ \hline
         Chronological & $\mathbf{OI}+\gabb \subseteq L$ & & \\ \hline
         Causally Non-Totally Vicious & $\mathbf{OI}$ & $\afterlogic \subseteq L $ & \\ \hline
         cNTV and Chronological & $\mathbf{OI}+\gabb \subseteq L$ & & \\ \hline
         Casual & & $\afterlogic+\gabb \subseteq L$ & \\ \hline
         Distinguishing & & & \\ \hline
         $\vdots$ & \multicolumn{3}{|c|}{$\vdots$} \\ \hline
         Globally hyperbolic & $\mathbf{OI}+\gabb \subseteq L$ & $\afterlogic +\gabb \subseteq L$ & $\mathbf{S4} \subseteq L$ \\ \hline
    \end{tabular}
    \label{tab:modalladder}
\end{table}

Table~\ref{tab:modalladder} shows the containment of the various modal logics for spacetimes and their various causal relations on different levels of the causal ladder.
As explored in Section~\ref{sec:separable}, the modal logic for two-dimensional spacetimes replaces $\afterlogic$ with $\aftertwologic$ in Table~\ref{tab:modalladder}.
The later levels of the causal ladder require new rules or additional operations to distinguish between their logics on the ladder.

\paragraph{Other effects on the logic.}
Note that spacetimes on the same rung of the causal ladder may nevertheless exhibit distinct modal logics.
For instance, contrast ordinary two-dimensional Minkowski space with the $2^{-}$-Minkowski spacetime: $\{(t,x)\in\mathbb{R}^2: t <0\}$, inheriting the same causal structure as in Example~\ref{example:minkowski space}.
Both are globally hyperbolic but the logics of those spacetimes under the causal (chronological) relation are $\mathbf{S4.2}$~\cite{Goldblatt1980} ($\mathbf{OI.2}$~\cite{Shapirovsky2002}) and $\mathbf{S4}$ ($\mathbf{OI}$) \cite[Lemma 3.6]{Shapirovsky2005} respectively.

Conversely, it is clear that the modal structure is not expressive enough to completely characterise the spacetime, since all Minkowski spaces of total dimension greater than three share the same modal logic.
More can perhaps be said using a multi-modal logic, where it is possible to quantify over causal and chronological curves, and then employing Malament's theorem~\cite{malament1977ClassContinuousTimelike}.

\begin{figure}[t]
    \centering
    \begin{subfigure}{.4\textwidth}
        \begin{tikzpicture}[scale=1]
    \tikzminkowski
    \node[label={[label distance=.05cm]-45:$w$}] (w) at (0,-1) {$\bullet$};
    \node[label={[label distance=.05cm]180:$x$}] (x) at (-.75,-.25) {};
    \node[label={[label distance=.05cm]0:$y$}] (y) at (.75,-.25) {};
    \node[label={[label distance=.05cm]45:$z$}] (z) at (0,.5) {$\bullet$};
    \draw (w.center) -- (x.center);
    \draw (w.center) -- (y.center);
    \draw (x.center) -- (z.center);
    \draw (y.center) -- (z.center);
\end{tikzpicture}
        \centering
        \caption{Minkowski spacetime}
    \end{subfigure}
    \hfill
    \begin{subfigure}{.4\textwidth}
        \begin{tikzpicture}[scale=1]
    \tikzminkowski
    \fill [pattern={dots}] (-\xmax,0) -- (-\xmax,\xmax) -- (\xmax+0.2,\xmax) -- (\xmax+0.2,0);
    \node[label={[label distance=.05cm]-45:$w$}] (w) at (0,-1) {$\bullet$};
    \node[label={[label distance=.05cm]180:$x$}] (x) at (-.75,-.25) {};
    \node[label={[label distance=.05cm]0:$y$}] (y) at (.75,-.25) {};
    \node (xl) at (-1,0) {};
    \node (xr) at (-.5,0) {};
    \node (yl) at (.5,0) {};
    \node (yr) at (1,0) {};
    \draw (w.center) -- (x.center);
    \draw (w.center) -- (y.center);
    \draw (x.center) -- (xl.center);
    \draw (x.center) -- (xr.center);
    \draw (y.center) -- (yl.center);
    \draw (y.center) -- (yr.center);
    \fill [pattern={Lines[angle=45,distance=3pt]}] (x.center) -- (xl.center) -- (xr.center);
    \fill [pattern={Lines[angle=45,distance=3pt]}] (y.center) -- (yl.center) -- (yr.center);
\end{tikzpicture}
        \centering
        \caption{$2^{-}$-Minkowski spacetime}
    \end{subfigure}
    \caption{Confluence on different spacetimes}
    \label{fig:diamond}
\end{figure}
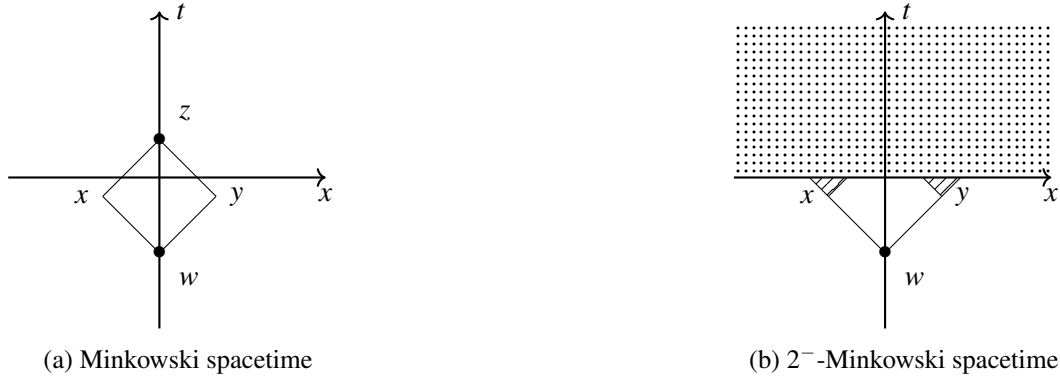

\paragraph{Multiple time dimensions.}
There is a notion of spacetimes with multiple time dimensions.
A $n$-spatial, $t$-temporal dimensional (or $(n,t)$-dimensional) spacetime can then be described as working on a $n+t$-dimensional manifold.
Our universe, $(3,1)$-dimensional, can be considered ``nice''.
By changing the number of spatial or temporal dimensions, the resulting universe could be considered either too simple or too unstable (not ``nice'').
For instance, Tegmark~\cite{Tegmark1997} provides various arguments for why observers cannot exist in $(n,t)$-dimensional universes where $(n, t) \neq (3,1)$, but the arguments made are not rigorous.

One fundamental question to ask is whether there is a fundamental logical reason that these universes are not ``nice''.
In Section~\ref{sec:separable:separability}, we showed that modal logic seems expressible enough to distinguish between $(1,1)$-dimensional and $(n,1)$-dimensional spacetimes ($n \geq 2$), but it is unlikely to distinguish any other (spatial) dimensional spacetime.
Extending our definition of $\necc$ and $\poss$ on the causal relations to higher temporal dimensions, or using multimodal logic to represent causal relations in different time dimensions, could provide the start towards a logical language to discuss what universes are ``nice'' or not.


\section*{Acknowledgements}
We thank Prakash Panangaden for helpful comments and feedback that helped improve our exposition, and for spotting a gap in the proof of~Section~\ref{sec:spacetime:after}.
This work has been partially funded by the French National Research Agency (ANR) within the framework of ``Plan France 2030'', under the research projects EPIQ ANR-22-PETQ-0007, HQI-Acquisition ANR-22-PNCQ-0001 and HQI-R\&D ANR-22-PNCQ-0002.

\AtEndDocument{
\bibliographystyle{eptcs}
\bibliography{bibliography}
}
\end{document}